\documentclass[12pt]{amsart}

\usepackage{cite}
\usepackage{xy}
\usepackage{graphicx}
\usepackage{amsfonts}
\usepackage{mathrsfs}
\usepackage{amssymb}
\usepackage{amsxtra}

\def\del{\partial}

\def\Z{\mathbb Z}

\def\R{\mathbb R}
\def\C{\mathbb C}
\def\Q{\mathbb Q}

\def\H{\mathscr{H}}

\def\eps{\epsilon}

\def\Tdeg{T_{\rm deg}}
\def\Sdeg{S_{\rm deg}}

\def\L{\mathscr L}
\def\bk{\mathbf{k}}

\def\V{\mathscr{V}}
\def\Bl{\mathscr Bl}
\def\bS{\mathbf{S}}
\def\bsigma{\boldsymbol\sigma}
\def\bx{\mathbf{x}}

\newtheorem{innercustomthm}{Theorem}
\newenvironment{customthm}[1]
  {\renewcommand\theinnercustomthm{#1}\innercustomthm}
  {\endinnercustomthm}

\newtheorem{thm}{Theorem}[section]
\newtheorem{lem}[thm]{Lemma}
\newtheorem{prop}[thm]{Proposition}
\newtheorem{cor}[thm]{Corollary}

\theoremstyle{definition}

\newtheorem{df}[thm]{Definition}
\newtheorem{qu}[thm]{Question}
\newtheorem{rmk}[thm]{Remark}

\accentedsymbol{\dbarG}{\Bar{\Bar{\Gamma}}}

\begin{document}

\title[Local and Global]
{Local models and global constraints for degeneracies and band crossings}

\author
[Ralph M.\ Kaufmann]{Ralph M.\ Kaufmann}
\email{rkaufman@math.purdue.edu}

\address{Department of Mathematics, Purdue University,
 West Lafayette, IN 47907 and MPIM Bonn, Vivatsgasse 11, 53111 Bonn, Germany}

\author
[Sergei Khlebnikov]{Sergei Khlebnikov}
\email{skhleb@physics.purdue.edu}

\address{Department of  Physics and Astronomy, Purdue University,
 West Lafayette, IN 47907}

\author
[Birgit Kaufmann]{Birgit Wehefritz--Kaufmann}
\email{ebkaufma@math.purdue.edu}

\address{Department of Mathematics and Department of Physics and Astronomy, Purdue University, West Lafayette, IN 47907}

\begin{abstract}
We study topological properties of families of Hamiltonians which may contain degenerate energy levels aka.\ band crossings.
The primary tool are Chern classes, Berry phases and slicing  by surfaces. To analyse the degenerate locus, we study local models. These give information
about the Chern classes and Berry phases. We then give global constraints for the topological invariants. This is an hitherto relatively unexplored subject.
The global constraints are more strict when incorporating symmetries such as time reversal symmetries. The results can also be used in the study of deformations.
We furthermore use these constraints to analyse
examples which include the Gyroid geometry, which exhibits Weyl points and triple crossings and the honeycomb geometry with its two Dirac points. 
\end{abstract}

\maketitle


\section*{Introduction}

%
%

Starting from considerations of families of Hamiltonians,
we give geometric and algebraic methods to study possibly higher band intersections;
these methods come from differential topology. Double crossings leading to
Dirac and Weyl points have been at the forefront of the investigations in the past few years. Our methods
extend beyond this, to triple and higher intersections. As an example, we analyse a triple intersection, stemming from a real world material  geometry (the Gyroid), and its deformation explicitly. Fabrication of the material---a nanowire network with the topology 
of a double gyroid---is described in \cite{Hillhouse} and numerical solutions to a wave
equation in such a network in \cite{Khlebnikov&Hillhouse}.

A family of Hamiltonians is a smooth map $H:T \to Herm(k)$ where $T$ is a smooth manifold and $Herm(k)$ is the space of Hermitian $k\times k$ matrices.
The smoothness is chosen for convenience, many arguments work on the $C^2$ level and some even on the topological level, i.e.\ for continuous families.
Such families arise naturally via Bloch theory  in condensed matter systems in $\mathbb{R}^d$ with translational symmetry $L\subset \R^d$, given a lattice $L\simeq \Z^d$.
Using Fourier transform, one obtains a family of Hamiltonians  $H(\bk)$ parameterised by quasi--momenta $\bk$ which are elements of the d--dimensional torus $T=T^d$.
We always keep this application in mind, but the methods are general.

There are two natural geometries associated to a family of Hamiltonians, the Eigenvalue and Eigenbundle geometry \cite{momentum}. The Eigenvalue geometry is  the cover of the parameter space by the energy levels. That is the cover $X\to T$ where $X\subset T\times \C^k$ is the subspace
whose points are $(\bk,spec(H(\bk)))$, where $spec(H(\bk))=\{\lambda_1,\dots,\lambda_k\}$ is the set of Eigenvalues of $H(\bk)$. Notice, that since $H(\bk)$ is Hermitian, the Eigenvalues are real and $X\subset T\times \R^k$.
This geometry was analysed in the general case in \cite{kkwk3}.
To give the Eigenbundle geometry, consider the (generalised) Bloch bundle  over $T:\Bl:=T\times {\mathbb C}^k\to T$. A physical state is a smooth section $s$ of the Bloch bundle and the Hilbert space of states
is given by all smooth sections $\H=\Gamma_{C^{\infty}}(T,\Bl)$, with the inner product induced from the standard Hermitian form on the fibers.
$H(\bk)$ acts on the fiber of $\Bl$ over $\bk$ simply as a matrix $H(\bk):\C^k\to\C^k$. This allows for a decomposition into Eigenbundles, which carries relevant information, since although
$\Bl$ is trivial,  its decomposition into Eigenbundles need not be.

Assume that $H(\bk)$ has a non--degenerate spectrum for each $\bk\in T$,
then the Bloch bundle decomposes into Eigenlinebundles $\Bl=\bigoplus_{i=1}^k \L_i$ and each of the line bundles can be non--trivial.
The non--triviality is measured by the first Chern classes $c_1(\L_i)$ and by the Berry phase, as we review below.
In general, $H(\bk)$ may be degenerate, and this more general situation is what we will analyse. This analysis was started in \cite{momentum} and we now add
local models and global aspects, such as symmetries, e.g.\ time reversal symmetry (TRS) and global topology to the mix.
This yields new global constraints and allows us in examples to completely characterise the Bloch bundle from local information.


For this analysis, we use Chern classes and thus $K$--theory.
The Chern classes can be computed using the Berry connection in the momentum space \cite{Berry,simon}.
This brings monopole charges and issues of topological stability into the picture and allows us to analyze deformations.

The paper is organised as follows:
After introducing the setup and reviewing the background, we present the main questions
  about local models and global constraints. In the second section, we  define and analyse local models. The models  which we will call ``of spin type" are especially important. The basic building blocks
go back to Berry's original examples \cite{Berry}, and Simon's \cite{simon} interpretation in a convenient formulation.
In the third section, we review the slicing technique for analysing 3d families and add generalisations.
We then introduce a new aspect in \S4, that is global constraints. We show what restrictions they entail, especially in the presence of symmetries. This partially answers a question of Berry, namely the singularities alone usually do not
determine the Eigenbundle geometry, but they do yield restrictions. In the presence of symmetries these may actually be enough to fully determine them.
This type of analysis can also be applied to the case of deformations. This leads to
the secondary question of stability. The global constraints allow for  complex topologies such as flux lines etc, but these are again restricted by any remaining symmetries.
There are minimal possible local singularities under deformations.

 These are realised in the specific examples that we analyse in \S5, such as  the Gyroid \cite{kkwk} and the honeycomb. Here the interesting new results are that
 \begin{enumerate}
 \item We give the local models and local Chern charges for the triple degeneracies and the double degeneracies (Weyl) points previously found in \cite{kkwk2}.
 \item The local data completely determine the global structure.
 \item The double Weyl points drift apart under deformations.
 \item   The triple points are of spin type and  have local Chern changes $-2,0,2$.
  \item Under deformations which preserve the time reversal symmetry (TRS), the triple points break up into four double or Weyl points each. This is the minimal possible dissolution of the triple points preserving TRS.
\end{enumerate}

\section{Setup and  Background}

\subsection{Eigenbundle geometry}
\label{evalpar}
We will follow \cite{momentum}.
As in the introduction consider a family of Hamiltonians $H:T\to Herm(k)$ and the trivial rank $k$ vector bundle
$\pi_{\Bl}:\Bl=T\times {\C}^k\to T$.  Let $T_{\rm deg}$ be the degenerate locus, i.e.\  $k\in T_{\rm deg}$ if and only if $H(k)$ has degenerate Eigenvalues. We will call these points critical or degenerate.
We will further assume that the components of $T_{\rm deg}$ are of at least codim $1$.
Let $T_0=T\setminus T_{\rm deg}$ the open complement, that is the locus where $H$ is non--degenerate.

The restriction $\Bl_0$ of $\Bl$ to $T_0$,  $\pi_{\Bl_0}:\Bl_0\to T_0$ then splits as a direct sum of line bundles
\begin{equation}
\Bl_0:=\bigoplus_{i=1}^k \L_i
 \end{equation}
 where $\L_i$ is the bundle of Eigenvectors of the i--the Eigenvalue.
  These are well defined by ordering the real Eigenvalues $\lambda_1<\dots<\lambda_k$.

This line bundle decomposition can usually not be extended to the degenerate locus, where level crossing, that is crossing of Eigenvalues happens.
For the whole space $T$, we can only  decompose
\begin{equation}
\Bl=\bigoplus_j \V_j
 \end{equation}
 where the $\V_j$ are rank $r_j$ vector bundles corresponding to the blocks of Eigenvalues that cross each other.
That is globally $\lambda_1\leq \dots \leq \lambda_{r_1}<\lambda_{r_1+1}\leq \dots \leq \lambda_{r_1+r_2}< \dots$, where $\sum_j r_j=k$.
Alternatively thinking about $H$ as an operator $H:\H\to \H$ this means that the Eigenbundles correspond to projectors commuting with $H$.

\subsubsection{Charges on the non-degenerate locus}
The main topological invariants of the Eigenbundle geometry are
the $K$--theory classes of the line bundles $\L_i$, see e.g.\ \cite{Atiyah},
in the $K$--theory of the non--degenerate locus: $[\L_i]\in K(T_0)$, which we call
the K--theoretic charges.

To these $K$-theoretical charges we obtain the more well--known  associated Chern classes
$\beta_i:=c_1(\L_i)\in H^2(T_{0})$ which we will call the cohomological charges, see e.g.\ \cite{Milnor}.
By general theory the total Chern class is given by $c(\Bl_0)=\prod_i (1+\beta_i)\in H^{\rm ev}(T_0)$, the even part of the cohomology.


One obtains numerical charges by pairing the cohomology valued Chern classes with homology classes.
By means of Chern--Weil theory this is usually implemented by integration of a differential form over a (sub)--manifold of the correct dimension. However, as these charges actually stem from the topological homology/cohomology pairing, which is defined over $\mathbb Z$, they are integers.

\subsubsection{Assumption}
In order to have this theory available and usable one needs certain ``niceness'' assumptions \cite{momentum}. Here,
we will consider the case where the components of $\Tdeg$ are contractible and are such that each component $T_c$ of $\Tdeg$ is  contained in the interior of a regular neighbourhood, that is a cell $N_{T_c}$, that is a sub--manifold homeomorphic to a closed ball of dimension $dim(T)$, and these submanifolds do not intersect. In this case $T_0$ is homotopy equivalent to $\bar T_0=T\setminus \amalg_{T_c} N^{int}_{T_c}$, where the sum is over the components of $\Tdeg$ and $N^{int}_{T_c}$ is the open interior of $N_{T_c}$.
If $T$ is a compact manifold then so will  be $\bar T_0$. If $T$ is a compact manifold, then $\bar T_0$ will be a manifold with boundary.

In the case that $\Tdeg$ is made up out of a discrete set of points this assumption is satisfied and  these submanifolds can  be taken to be balls centered at the degenerate points. For a more general setup see \cite{momentum}.
This assumption is also satisfied if the components of $\Tdeg$ have finitely many contractible components.

Using this assumption, we can equivalently consider the charges for $T_0$ to lie  in $H^*(\bar T_0)\simeq H^*(T_0)$ and in $K(\bar T_0)\simeq K(T_0)$.
With this assumption the results we state are not in their most general form, but it relieves us from too much technical detail.
In concrete situations, it is easily checked if the results can be extended.

\begin{rmk}
We can also consider the total Chern classes $c(\V_i)\in H^{ev}(T)$.
If this has usable information depends on the family. If for instance all bands cross,
we only get $\Bl=\V_1$ which is trivial and hence $c(\V_1)=1$.
\end{rmk}

\begin{rmk}
We have  assumed that the Hamiltonians are generically non--degenerate. Technically, it is sufficient to assume that the ranks of the Eigenbundles are generically constant. In this case, the singular locus is where the rank jumps up and instead of line bundles over the non-degenerate locus one will have vector bundles $\V_i$ and total Chern classes $c(\V_i)$. This is important for the case in which every level is doubly degenerate, such as for instance caused by a spin symmetry where the $\L_i$ are replaced with vector bundles $\V_i$ of rank $2$. We will deal with this case in the future.
\end{rmk}

\begin{rmk}
Notice that the charges are trivial if $T_0$ has vanishing second cohomology (e.g. if $T_0$ is 2--connected).
In that case the Chern classes $\beta_i$ vanish and the line bundles $[\L_i]$ are trivializable.
This is the case in some examples, notably the honeycomb. Another consequence of this triviality is that the associated points of degeneracy are not topologically stable.
The two--torus or the two--sphere do however have non--vanishing $H^2$ and thus are prime candidates
to carry non--trivial first Chern classes and hence non--trivial bundles
with non--trivial Berry phases.
\end{rmk}

\subsubsection{Scalar Topological Charges}

To code this information into measurable numbers, one needs to pair
the cohomological charges with homology classes. In the differentiable setting
this corresponds to the integral of the curvature form for any connection over
a cycle of the correct degree. The set of {\em all} such numbers on a set of generators of homology of $T_0$
then determines the cohomological charges as functions on homology.
If we use at least $\Q$ coefficients (usually in physics one takes $\R$ or $\C$ to represent everything by forms and integrals), this in turn completely fixes the line bundles as given by the Chern isomorphism theorem and the classification theorem for line bundles, see e.g.\ \cite{husemoeller}.

For these considerations, it is easier to assume that we are dealing with oriented manifolds. If we furthermore have a differentiable structure, we know that we can evaluate Chern classes
by using Chern--Weil theory. E.g. if $A$ is a connection form for the line bundle, we can represent
the first Chern--class by the curvature form $\Omega= dA+\frac{1}{2}A\wedge A$.

Using other even dimensional homology cycles, we can also extract some information of the corresponding combinations of first Chern classes according
to the usual formulas for the full Chern class in terms of (virtual) line bundles; see e.g.\ \cite{Hirzebruch}.

%
%
%
%
%

\subsubsection{Berry phase/connection}
\label{Berrypar}

Following Berry \cite{Berry} we can use the connection $A_{Berry}$ provided by adiabatic transport for the line bundles $\L_i$.
It was Berry's insight that this connection is indeed not always trivial
 and produces the so--called Berry phase as a possible monodromy.
 In particular, if $C$ is a closed circuit and $|\psi\rangle$ is a state then
 adiabatically moving $|\psi\rangle$ around $C$ may introduce an extra geometric phase $e^{i\gamma(C)}$. It is important to note that the quantity $\gamma(C)$ is only defined up to multiples of $2\pi$.

The phase can be computed using the so--called Berry connection and Stokes theorem.
For this one considers a surface whose boundary is $C$ and then computes the integral of the connection over the surface to obtain $\gamma(C)$, see below for an example. What is important to point out here, is that the computation does depend on the chosen surface, but only up to adding multiples of $2\pi$.

 Simon \cite{simon} noticed that integrating this connection over a closed surface $S$ computes exactly the first Chern class $c_1(\L_i)$ of the line bundle $\L_i$ paired with $S$.

The usual Chern--Weil form for any choice of connection is given by an expression in the curvature for a choice of a connection \cite{ChernSimons}. One such choice for a line bundle is the Berry connection.  Stokes' theorem then links the computation of the Berry phase to the integral of the vector field $V$ given by th curvature form  over a bounding surface $\iint_S  V dS= \oint_C A_{Berry} dr$. This was related \cite{BottChern} to the first Chern class by changing the representation of $V$  using the Bott and Chern connection and realising that in this form $V$ satisfies $\frac{1}{2\pi}\int_SV_m dS=c_1$.
These computations are linked to the Chern-Simons forms $Q_{2l-1}$
through the fundamental relation that $dQ_{2l-1}=ch_{2l}$ where $ch_{2l}$ is the degree $2l$ part of the Chern character \cite{ChernSimons}. This will be further explored elsewhere.
\subsection{Geometry of $Herm(k)$}

\subsubsection{Full Family vs.\ concrete families}
Traditionally, the ``ge\-ne\-ric scenario'' has been of interest. This is a (generic subset of) the tautological family $T=Herm(k)$ and $H=id:Herm(k)\to Herm(k)$ is simply the identity map. The study of the full family  goes back to \cite{vNW}.
As our analysis deals with general variations, that can be non--generic in
the above sense ---and sometimes even have to be due to the presence of extra symmetries---  the results about the generic case merely provide expectations
which may or may not hold in the given situation.

The most prominent results on the generic geometry of $Herm(k)$ were already obtained in \cite{vNW}. Here one can find the co--dimensions of the strata of degenerate Eigenvalues, basically by a dimension count.
A particularly well known fact is that generically the locus of degenerate Eigenvalues, that is Eigenvalues of multiplicity $>1$, is of codimension 3 \cite{vNW}. Thus $3$ is the expected codimension, but in a given variation this may or may not be the actual codimension, and we have examples of both types of behaviour. For the real situation one finds generically that the codimension is $2$. The analysis of the geometry of the tautological family was carried further in \cite{arnold}, where a filtration was introduced. Arnold \cite{arnold} studied this filtration and that study has been continued in \cite{agrachev}.
 This sphere bounds a $k^2-1$--dimensional ball to which the family naturally extends. This has a maximally degenerate point at zero. In general, the Hamiltonians on the sphere can also be degenerate.

\subsection{Effective sphere families}
 For these and other discussions it is  convenient to mod out the $k^2$-real-dimensional vector space $Herm(k)$ by translations and dilatations as
shifting (adding constant scalar matrices) or scaling (multiplying by non--zero constants) the spectrum or scaling it does not change the topology of the situation.
Modding out by the translations means that we can restrict to traceless matrices and modding out by dilatations means that after choosing a basis we can scale the corresponding vectors to be of norm $1$, unless we are dealing with the $0$ matrix; see below for the case of $2\times 2$ matrices.
The quotient space of the space of {\em non--scalar} Hermitian matrices under the simultaneous action, which is naturally identified with the co-invariants, is then a $k^2-2$--dimensional sphere.
 This sphere then has a filtration by pieces $F_p$ consisting of those points where the first $p$ Eigenvalues are equal.

\subsubsection{$Herm(2)$}
\label{2x2case}
In the special case of $2\times 2$ Hermitian matrices it is well known that the Pauli matrices
$$\sigma_0=\left(\begin{matrix}1&0\\0&1\end{matrix}\right),
\sigma_x=\left(\begin{matrix}0&1\\1&0\end{matrix}\right),
\sigma_y=\left(\begin{matrix}0&-i\\i&0\end{matrix}\right),
\sigma_z=\left(\begin{matrix}1&0\\0&-1\end{matrix}\right)
$$
 form a basis for the 4--dimensional space of Hermitian matrices. The traceless matrices are spanned by $\sigma_x,\sigma_y,\sigma_z$ and are of the form
 \begin{equation}
 \label{Weylfameq}
 {\mathbf x}\cdot{\boldsymbol\sigma}=x\sigma_x+y\sigma_y+z\sigma_z=
 \left(\begin{matrix}z&x-iy\\x+iy&-z\end{matrix}\right),
 \end{equation}
Restricting the $2^2-2=2$ dimensional sphere, restricts $\bf x$ to lie on $S^2\subset {\mathbb R}^3$, i.e.\ $x^2+y^2+z^2=1$. This family is entirely non--degenerate. Notice that this $S^2$ is centered around the zero--matrix and the extension of the family to the 3--ball $B^3$ has an isolated degenerate point at $\bf 0$.
The ball family is the local model for a doubly degenerate point, aka.\ Weyl point.

\subsubsection{$Herm(3)$} In the case of $k=3$ the traceless Hermitian matrices are spanned by the $8$ Gell--Mann matrices $\lambda_i,i=1,\dots, 8$. Modding out by dilatations,  one can restrict to an $S^7\subset \mathbb{R}^8$.
Here the center of the sphere is again at the origin and is a 3--fold degeneracy. The family has 2--fold degeneracies on the sphere.
Progress on the full analysis of the family, especially on the degenerate part $S^7_{\rm deg}$ has been made in \cite{GellMann,SimonQuaternion} where this is linked to an $S^4$ which naturally supports second Chern classes.

\subsection{Local Models}
Locally the behaviour near a particular point is given by the family restricted to a regular neighbourhood, that is locally the families are described by families on a ball.
Thus we define {\it a basic local model} to be a germ of a diffeomorphism class of maps $B^n\to Herm(k)$,
where we identify two classes if one is contained in the other by the restriction to a smaller ball with the same center. Two germs are equivalent if they result from each other by conjugation by a unitary linear transformation.

A local model is the direct sum of basic local models. A basic model is called simple if the Bloch bundle does not split into a direct sum of subbundles.

\subsection{Local charges}
\label{chargepar}
For each component $T_c$ of $\Tdeg$, we can consider the submanifold $\del N_{T_c}$ which is homeomorphic to a sphere of dimension $dim(T)-1$ and consider the restriction of $Q_c$ to it.

We define the {\em local charge} of that component to be
$$Q_{T_c}:=\int_{\del N_{T_c}} i^* (Q_c)$$
where $i:\del N_{T_c}\to T$  is the inclusion.

This is of course only interesting if $T$ is odd dimensional and hence $\del N_{T_c}$ are even dimensional spheres. In the even dimensional case, that is odd dimension of $\del N_{T_c}$ we can still consider Chern--Simons classes and Berry phases.

The charges are invariant under equivalences and homotopies in the appropriate sense.

\subsection{Main Questions: Local models and Global Constraints}

This  immediately begs the following questions, already raised in Berry's original article \cite{Berry}.
\begin{qu}\mbox{ }
\begin{enumerate}
\item Is it possible to classify the local models?
\item Is it possible to classify the local charges?
\item How are these constrained by the geometry of the base/family?
\item How do points/the degenerate locus behave under deformation?
\end{enumerate}
\end{qu}
We will address these questions below.
The classification is possible in certain cases. I.e.\ for instance for $k=2$ and an
isolated regular singularity. In this case, it is just the family given in \S\ref{2x2case}.

The first will lead us to consider
local models and the latter to introduce global restrictions.
The surprising fact is that sometimes these are enough to determine the spectrum.
For positive results see Theorem \ref{constrainthm}.

As to the last question. Indeed the first expectation that the isolated critical points behave
like monopoles is not quite  correct,  as already Berry noted.
 First, the Chern charge does not depend on the total spin, see \S\ref{multcrosspar} for details, and secondly
under general deformations the degenerate locus can  split, deform and smear out, see below.
What is, however, true
is that the local charges have to be preserved, in the sense that if they split or create singularity
loci of higher dimension, the total local charges in the sense of \S\ref{chargepar} have to be preserved.
Here one has to take $N_{T_c}$ large enough to contain all the components created when deforming the degenerate locus $T_c$.

\subsection{Deformations and Topological stability}
\label{stablepar}
Having non--vanishing topological charges produces topological stability. If we perturb the Hamiltonian slightly by adding a small perturbation term $\lambda H_1$ and continuously vary $\lambda$ starting at $0$, then  $\Tdeg$  and thus $T_0$ does not move much ---for instance as  submanifolds  of $T\times R$ where we keep the base $T$ constant. This follows for instance from the description of the Eigenvalue geometry using the characteristic map.  The Eigenbundles over $T$ also vary continuously and hence so do their Chern classes. Since these are defined over $\Z$ they are actually locally constant, so that all the non--vanishing charges, scalar, K-theoretic or cohomological, must be preserved.
That is, the total local charges will be preserved on $\bar T_0$ as long as we cut away enough, that is make $N_{T_c}$ large enough.

However, there is no guarantee that the local charges are ``carried'' by single
points and that the number of these is preserved. We will give a concrete example, where one triple degenerate point decomposes into four double points.
Likewise points could possibly degenerate into lines. This is however
not generic. The opposite phenomenon, i.e.\ contraction of a dimension 2 or higher
locus to a point is certainly possible. All these deformation have to preserve the local charges. This is why Weyl points are of interest. If there is a non--trivial charge associated to them, they cannot decay.

\section{Local Models}

\subsection{Local models from the Eigenvalue geometry}

In \cite{kkwk3}, we proved that in general the fibers over points of
$\Tdeg$ have singularities pulled back from the singular locus, aka.\ the swallowtail  of the $A_{k-1}$ singularity, and
are classified by types $(A_{k_1},\dots, A_{k_l})$ with $\sum_j(k_j+1)=k$.
In particular,  locally the Eigenvalue geometry is pulled back from the unfolding of the $A_{k-1}$ singularity under the so--called characteristic map $\Xi:T\to \mathbb{C}^{k-1}$ of miniversal unfolding $M\to \mathbb{C}^{k-1}$ of the $A_{k-1}$ singularities, see {\it loc cit.}.
These are known due to Grothendieck \cite{grothendieck} to be stratified with the strata corresponding to the possibilities to delete vertices (and the incident edges), whence the classification. Deformation of the family deforms the map $\Xi$ and with it the crossings and $\Tdeg$ which is the inverse image of the swallowtail under $\Xi$.

What was not stressed in \cite{kkwk3} is that for Hamiltonian families $(A_{k_1},\dots, A_{k_l})$ is actually an ordered set, since then everything is defined over $\R$. It is ordered by
the values of the Eigenvalues as discussed in \S \ref{evalpar}. We will start with the lowest Eigenvalue first.

 For instance,  for the Gyroid, for which $k=4$ we found
two triple crossings with types $(A_2,A_0)$  and $(A_0,A_2)$  and  two double Weyl crossings of type $(A_1,A_1)$.

\subsection{A simple local model for the Eigenbundle geometry of  an isolated $n$--fold degeneracy on a $3$d base}
\label{multcrosspar}
In particular, there are 3-parameter models for all isolated normal singular crossings of $n$ Eigenvalues, that is isolated
$A_{n-1}$ type singularities.  For a double crossing this local model is essentially unique, see Corollary \ref{doublecor}, for higher crossings there might be other possible models.
These were already explored by Berry \cite{Berry} and can also be found in \cite{simon}.
In particular, they exhibit  an isolated point in $T_{\rm deg}$ with maximal degeneracy and  the degeneracy is lifted to first order in each direction, which is what is called a ``normal singular'' in \cite{simon}.

Generically such local models are expected to appear
 when $T$ is of dimension $3$ as the degenerate locus should be of codimension $3$ and hence consist of isolated points.

\noindent{\sc Conventions:}
Fix an integer $S$ and let $T=\R^3$.
 Consider the $S$--dimensional spin representation of $su_2$ that is
given by the collection of matrices $\bS=(S_x,S_y,S_z)$ which act on $\C^{S}$ and satisfy the usual commutator relations
$[S_x,S_y]=iS_z$, $[S_y,S_z]=iS_x$, $[S_z,S_x]=iS_y$.
The possibly half--integer $s$ is defined via  $S=2s+1$ and is called the spin of the representation.


Notice that $S_z$ is diagonalisable with Eigenvalues
$S_m:m=-s, -s+1, \dots, s-1, s$ where $m$ is integer or half--integer depending on whether $S$ is odd or even.
Consider the family of traceless Hamiltonians
\begin{equation}
H({\bf x}):={\bf x}\cdot {\bS}=xS_x+yS_y+zS_z
\end{equation}
on $\R^3$. This is rotationally symmetric and has only one critical point at $\bf 0$. It is totally degenerate, that is all $2s+1$ bands cross.
Thus on $T_0=\R^3\setminus \{{\bf 0}\}$ the family has no critical points and the line bundles $\L_m, m=-s,\dots, s$ corresponding to the Eigenvalues above
are well defined over all of $T_0$. If we restrict the family to the homotopic family $S^2\subset T_0$, we get K-theoric, cohomological,  and numerical Chern charges. Once an orientation is chosen,
these all carry the same information since the choice of orientation establishes an isomorphism $H^2(S^2)\simeq\Z$ and the reduced $K$-theory of $S^2$ is also identified with $\Z$.
One can calculate \cite{Berry,simon} that the Chern charges are
$$\int_{S^2}c_1(\L_m)=2m$$
Note that this is independent of the value of $S$.
Here the orientation is the usual orientation of $S^2\subset \mathbb{R}^3$.

Notice that a reversal of orientation will change the isomorphism sending $1$ to $-1$ and hence flip the sign.

Moreover, according to Berry \cite{Berry}, the Berry phase for the bundle $\L_m$ along a closed circuit $C$ is proportional to the solid angle subtended over a surface $S$ which has $C$ as a boundary.
More precisely,
\begin{equation}
\gamma_m(C)= m \iint_S d\Omega
\label{berryeq}
\end{equation}
where $d\Omega$ is the solid angle two--form $sin(\theta)d\phi d\theta$.
It is important that $\int_S d\Omega$ depends on the choice of $D$ and is only well defined up to a change of $4\pi$.
In case that $C$ is the equator counterclockwise and $D$ is the upper hemisphere, \eqref{berryeq} becomes
\begin{equation}
\label{berry2eq}
\gamma_m(C)= m 2\pi
\end{equation}
In the case of spin $\frac{1}{2}$ this will be $\pm \pi$. In the 3--band system for spin 1, this will take
values $-2\pi,0,2\pi$ depending on $m$.
Choosing the lower hemisphere would result in a difference of $m4\pi=2\pi c_1(m)$ which is always an integer mod $2\pi$.

It is important to note that ``spin'' here refers to the particular type of Hamiltonian and does not have to coincide with physical spin.
%
%
%
%

\subsection{Spin--Type Models and Their Charges}

\begin{df}
 We say that  an isolated point  $\bk_0 \in \Tdeg$ is of spin type  $(s_1,\dots, s_l)$,  if it is of singularity type $(A_{2s_1},\dots,A_{2s_l})$ and
 there is a linear isomorphism $L_{\phi_j}$   for each $A_{k_j}$ singularity in the Eigenvalues  to first order perturbation  theory
$P_j[H(\bk_0+\bx)-H(\bk_0)]P_j={\bf a_j}\bx id+ L_{\phi_j}(\bx) \cdot \bS +O(\bx^2)$ where ${\bf a}_j$ is a vector, $\bS=(S_x,S_y,S_z)$ is a spin $s_j$
representation of $su(2)$ and $P_j$ is the projector onto the degenerate Eigenspace of the $2s_j+1$ fold crossing.
\end{df}

This definition is a bit technical, but practical.
Examples for the Gyroid, see below for details, are points of spin type $(0,1),(1,0)$ and $(\frac{1}{2},\frac{1}{2})$ where $3$ of $4$ bands cross, or $2$ and $2$ bands cross.

If we subtracted the trace  to be in the case of traceless matrices, then we get a nice equivalent homotopy characterisation.

\begin{thm}
An isolated  point  ${\bf k}_0\in \Tdeg$ for a  3--dimensional family
is of   spin type $(s_1,\dots,s_l)$  if and only if the local model of $H-Tr(H)$ at $\bk_0$ is homotopic, through a homotopy of families with only one isolated critical point,
to a direct sum of Hamiltonians of the corresponding spin Hamiltonians of the form of \S\ref{multcrosspar}.
That is,
there is a regular closed neighbourhood $V$ of $\bk_0$ and  diffeomorphisms $\phi_j:V\to B^3$ such that on $V$:
$H(\bk)-Tr(H(\bk))$ is homotopic to $\phi_1(\bk) \cdot \bS_{s_1}\oplus \dots \oplus \phi_l(\bk) \cdot \bS_{s_l}$ via a homotopy of families only degenerate at $\bf 0$.
\end{thm}

\begin{proof}
If there are such a diffeomorphisms and a homotopy then expanding $P_jH(\phi_j(\bk))P_j$ to first order, we see that we have a family homotopic to spin type where $L_{\phi_j}$ is equal to the Jacobian of $\phi_j$.

If the point $\bk_0$ is of spin type, consider the first order perturbation theory as above. Now using a unitary transform $U$ to diagonalise $H(\bk_0)$ we have
that  $U^\dagger H(\bk_0+\bx)U=tr(H(\bk_0+\bx))+\tilde H$ where $\tilde H= U^\dagger [H(\bk_0)-trH(\bk_0)]U+ \tilde H_1(\bx)+O(x^2)$  with $\tilde H_1(\bx)$ a traceless matrix and
$\bigoplus_{j=1}^k P_j\tilde H_1(\bx)P_j=\bigoplus_{j=1}^k L_{\phi_j}(\bx)\bS_j$ where $P_j$ is the projector onto the degenerate Eigenspace corresponding to the degenerate Eigenvalues $\lambda_j$. We can now homotope unwanted terms away  in three steps. First we homotope any higher order terms  by scaling them to zero. Since to leading order the spin Hamiltonians resolve the degeneracies choosing a small enough neighbourhood, this can be done through families only degenerate at $\bf 0$.
Second, we can homotope away the traceless diagonal term  $U^\dagger [H(\bk_0)-trH(\bk_0)]U$ by restricting the family to a neighbourhood of size $\eps<
\frac{1}{2}\frac{max(|\lambda_i-\lambda_k|)}{max |s_i|}$. The homotopy is simply given by $(1-t)U^\dagger [H(\bk_0)-trH(\bk_0)]U +\tilde H_1$.
 Since we are in a neighbourhood of radius less than $\eps$, the Eigenvalues will not cross during the homotopy and $\bf 0$ will remain the only non--degenerate point.
In the last step we homotope away all unwanted coefficients
 of the matrix $\tilde H_1$ outside the blocks corresponding to the projections.
 This can be done by the  homotopy $(1-t)[\tilde H_1-\bigoplus_{j=1}^k P_j\tilde H_1(\bx)P_j]+\bigoplus_{j=1}^k P_j\tilde H_1(\bx)P_j$, since the degeneracies are resolved to first order by
 restricting to a smaller neighbourhood if necessary.
%
%
\end{proof}

\begin{cor}
If $\bk_0$ is an isolated point of $\Tdeg$ of spin type   $(s_1,\dots, s_l)$,
 then the local charge of $\L_m$ where $m=-s_j,\dots, s_j$ corresponding to the given summand $j$ is $sign(\phi_j)2m$ where $sign(\phi_j)$ is the sign of the determinant of $L_{\phi_j}$.
\end{cor}
The sign $sign(det(L_{\phi_j}))$ is independent of $m$ and will be called the {\em chirality}.
\begin{proof}
Since the homotopy preserves the non--degeneracy on the  $S^2$ boundary of the ball and the  Chern classes are homotopy invariant, we have that $c_1(\tilde \L_m)=c_1(L_{\phi_j}^*( \L_m))=L_{\phi_j}^*c_1(\L_m)=sign(det(L_{\phi_j}))2m$ where $\tilde \L_m$ corresponds to the line bundle of $H(\bk)$ and $\L_m$ is the line bundle of \S\ref{multcrosspar}.
The last equation comes from the fact that the degree of the morphism $L_{\phi_j}$  is given by the sign of the determinant, that is $+1$ if $L_{\phi_j}$ is orientation preserving on the ambient $\R^3$ and $-1$ if it is orientation reversing.
\end{proof}

As a corollary, we obtain  a result which can be found in  \cite{simon}:
\begin{cor}
\label{doublecor}
 In particular, in the case of a double crossing, that is a singularity of the type $A_1$ without any additional assumption,
 $P\tilde H_1P$ is a traceless $2\times 2$ matrix and hence $P\tilde H_1P$
is always of the form $L_{\phi_j}(\bx)\bS$ and hence of spin type. If the $n+1$-th and $n$-th band cross then the local charges are $sign(det(L_{\phi_j}))$.
\end{cor}
\label{negcor}
The following corollary is also very useful.
\begin{cor}
If $H(\bk)$ is of spin--type $(s_1,\dots,s_l)$ at $\bk_0$, then $-H(\bk)$ is of spin--type $(s_l,\dots,s_1)$
with the opposite chirality.
\end{cor}

\begin{proof}
The  $j$th Eigenvalue of $-H$ is the $l-j$th Eigenvalue of  $P_j[-H(\bk_0+\bx)-(-H(\bk_0))]P_j=-P_{l-j}[H(\bk_0+\bx)-H(\bk_0)]P_{l-j}=
-{\bf a}\bx + -L_{\phi}(\bx) \cdot \bS +O(\bx^2)$ and  the chirality is $sign(det(-L_{\phi_j}))=-sign(det(L_{\phi_j}))$.
\end{proof}

\subsubsection{Berry phases in 2d}
The above calculations can also be truncated to 2d, that is 2d subfamilies in the 3d family $\bx\cdot \bS$. We say that a 2d isolated singular point in $\Tdeg$ is of spin type $(s_1,\dots,s_k)$ if
to first order deformation theory the local family is a 2d subfamily of a 3d subfamily of spin type $(s_1,\dots,s_k)$.

The intersection of an embedded 2d subfamily with a Dirac point (i.e.\ it contains $(0,0,0)$), with a small $S^2$ around $(0,0,0)$ will be a non--empty  closed curve  and the monodromy is given by \eqref{berry2eq}.

A common type is the equatorial subfamily $z=0$ that is $xS_x+yS_y$. We define the chirality analogously as $sign(det(L_{\phi_j}))$ where now $L_{\phi_j}$ is a $2\times 2$ matrix.

\begin{lem}
For an equatorial subfamily the value of $\gamma(C)$ defined by the upper hemispheres is given by
$\gamma(C)=\pm 2\pi m$, that is \eqref{berry2eq}, with the additional sign given by the chirality.  \qed
\end{lem}
\begin{rmk}
Notice that the sign of $\gamma(C)$ depends on the choice of the upper hemisphere as spanning surface, cf.\ \S\ref{Berrypar}. The physical Berry phase does not depend on this. It is however interesting to see the different chiralities that appear in one family, e.g.\ that of graphene, cf.\ \S\ref{Graphenesec}.
\end{rmk}

\section{Topological Charges and Slicing}
 To obtain effective global constraints, we recall the technique of slicing, cf.\ e.g.\ \cite{momentum,Xu&al}.
The idea is that we can evaluate the first Chern class of a line bundle with a connection on a 2--dimensional submanifold by pulling back, i.e.\ restricting, the line bundle to the surface and integrating the pulled--back curvature form of the connection over the surface. Explicitly, if $\Sigma$ is an oriented compact surface and $i:\Sigma\to T_0$ is an embedding, then

\begin{equation}
\label{surfacechargeeq}
Q_{\Sigma,i}:=\int_{\Sigma} i^*c_1(\L_i)=\langle c_1(\L_i), i_*([\Sigma])\rangle
\end{equation}
where $\langle \;,\;\rangle$ is the standard pairing between cohomology and homology.
Notice that by the results of Thom \cite{Thom} all second  homology classes are of this type, even over $\Z$ and hence representing all cohomology classes in this way, the numerical Charges $Q_{\Sigma,i}$ fix the cohomological charge.

\subsection{3--dimensional torus models}

For concreteness and with the applications in mind, we will now recall the case where $T=T^3$ and restrict the charges to those coming from $c_1$; for a more general discussion, see \cite{momentum}.
We represent  $T^3$ as the cube $[0,2\pi]^3$ with periodic boundary conditions $0\sim2\pi$. In particular, we will write $-t$ for $2\pi-t$.
We can then consider the embedding of $T^2$ into $T^3$ at ``height $t$''. That is, the slicing with respect to the $z$ coordinate is defined by
\begin{equation}
\phi_t(\theta_1,\theta_2)=(\theta_1,\theta_2,t)
\end{equation}
and the other two coordinate slicings are defined analogously.

Given a family $H:T^3\to Herm(k)$ we obtain the functions

\begin{equation}
\chi_i(t):=\int_{T^2}\phi_t^*c_1(\L_i), \quad {i=1\cdots k}
\label{slicefuntions}
\end{equation}
For all $t$ such that $\phi_t(T^2)\subset T_0$, that is, it  does not contain any degenerate points.

We now assume that  $\Tdeg$ is regular, that is
its components are finitely many contractible sub--manifolds. This notion is  less restrictive than the one used in \cite{momentum}.
This implies that the coordinate projections of $\Tdeg$ are finitely many points and intervals in each $S^1$.
We will also assume that they are in generic position with respect to an identification $T^3\simeq S^1\times S^1\times S^1$. This means that all their coordinate projections $\pi_k:T^3\to S^1, j=1,2,3$ for any two components are non--intersecting. We can always obtain  generic position by using a diffeomorphism homotopic to the identity.

Notice that in this situation, the slicing only gives a finite set of numbers for each Eigenbundle, since the integral over the Chern--class is invariant under homotopy and hence the $\chi_i$ are locally constant and constant in the components $S^1\setminus \pi_k(\Tdeg)$.

In this case the following generalisation of \cite[Theorem 3.13]{momentum} applies:

\begin{thm}
\label{slicethm}
For a smooth variation with base $T^3$ and regular $\Tdeg$, which we may assume to be in generic position, the functions obtained from the slicing method corresponding to all three coordinate projections completely determine the $K$--theoretic charges and hence the line bundles $\L_i$ up to isomorphism.
\end{thm}
\begin{proof}
The main ingredient in the proof was a CW complex obtained by a grid given choosing points in the components of the $S^1\setminus \pi_k(\Tdeg)$. In this grid, each 3--cell contains one component of the degenerate locus. Since this is contractible, in computing the homology, we are reduced to the case of \cite[Theorem 3.13]{momentum}. In particular contracting all the 3--cells to their boundary, we obtain a CW model for $T_0$ and the theorem follows as in \cite{momentum}.
\end{proof}

%

\subsection{Jumps and local charges}

The  locus of discontinuity for each function is a closed set consisting of isolated points and intervals.
For each component $T_c\in \Tdeg$ we  define the jump at $T_c$ as follows.
If $\pi_k(T_c)$ is  a critical interval $I_c=[t_c^1,t_c^2]$ then we set
\begin{equation}
j_i(T_c)=\chi_i(t^2_c+\eps)-\chi_i(t^1_c-\eps)
\end{equation}
If  $T_c$ is an isolated point and $t_c=\pi_k(T_c)$, then $t^1_c=t^2_c=t_c$ and we set
\begin{equation}
j_i(t_c):=j_i(T_c)=\chi_i(t_c+\eps)-\chi_i(t_c-\eps) \text{ for small $\eps$}.
\end{equation}

\begin{rmk}
The significance of these jumps is as follows:
 For this, consider a regular neighborhood $N_{T_C}$
 of $T_c$ and choose  $\eps/2$ be such that $\pi_k(N){T_c}\subset (t_c^1-\eps/2,t_c^2+\eps/2)$, where $\pi_k$ is the projection under consideration. Let
  $S^2_{T_c}$ be a boundary part of this neighborhood which is diffeomorphic to a sphere and  let $B_{T_c}$ be the open part which is diffeomorphic to the ball inside of the sphere. Now consider the 3--manifold $T_{slice}$ between two slices, that is e.g. $\{(\phi_1,\phi_2,t):t\in [t-\eps, t+\eps]\}\cap T\setminus B_{T_c}$. Then since Chern forms are closed: $0=\int_{T_{slice}} dc_1(\L_i)=-\int_{S^2_{T_c}}c_1(\L_i)+\int_{T^2}\phi^*_{t_c^2+\eps}c_1(\L_i)-\int_{T^2}\phi^*_{t_c^1-\eps}c_1(\L_i)$
by Stokes and hence $j_i(T_c)=\chi_i(t^2_c+\eps)-\chi_i(t^1_c+\eps)=\int_{S^2_{T_c}}c_1(\L_i)$.
In other words {\em the jumps equal the local charges}.
Thus, if we know the local models, we have the information about the local charges and hence in the slicing method, we know the jumps.
\end{rmk}

\begin{rmk}
If a slice $\phi_t$ cuts  $\Tdeg$ in isolated points, we can use Berry phase analysis.
If one knows for instance we have equatorial 2d singularities, one can determine the chiralities and Berry phases around these points.
This provides an alternative approach for the analysis.
\end{rmk}
%

\section{Global constraints}
Fix a system with base $T^3$ and a slicing in generic position with respect to the projection $\pi_k:T^3\to S^1$. Let $t\in S^1$ be the slicing parameter and $\pi_k(\Tdeg)=\Sdeg\subset S^1$ be the locus of
points such that $\phi_t:T^2\to T^3$ hits the critical locus $T_{crit}$.

\subsection{Global Constraints for the Slicing Charges}

\begin{thm}[Part 1]
\label{constrainthm} The periodic functions $\chi_i$  defined in eq.\ (\ref{slicefuntions}) satisfy the following:
\begin{enumerate}
\item They are locally constant on $S^1\setminus \Sdeg$, moreover they are stepfunctions with integer values.
\item $\sum_1^k\chi_i\equiv 0$.
\item For every component $T_c$ of $\Tdeg$: $\sum_1^k j_i(T_c)=0$
\item $\sum_{T_c\in \Tdeg} j_i(T_c)=0$, where $T_c$ runs over the isolated critical points and a choice of point for
each of the critical intervals.
\item The jumps at an isolated double crossing are given by $j_i(t_c)=\pm\sigma$ (the sign is determined by Corollary \ref{doublecor}) and the jumps at a multiple crossing of local type ${\bf x} \cdot {\bf L}$ are given by $j_m(t_c)=\pm m$ (according to \S\ref{multcrosspar}).
 \end{enumerate}
\end{thm}
\begin{proof}
The first statement is straightforward, the second follows from the fact that the $\bigoplus_i \L_i$ is a trivial line bundle and implies (3). The fourth statement is the periodicity of the functions $\chi$ and the last statement follows by Stokes for a small sphere around the isolated critical point.
\end{proof}

\begin{cor}
On $T^3$ for $k=2$ there are no families with a single regular critical (aka.\ Weyl) point. If $\Tdeg$ only has regular isolated  points, such points appear in pairs with opposite chirality.
\end{cor}

\begin{proof}
If there was only one regular critical point then there would only be one jump by $\pm 1$  for the functions $\chi_i, i=1,2$ and this would violate (3).
In order to obtain $0$ as the total jump, one has to have as many jumps up as down, which proves the second statement.
\end{proof}


\subsection{Global Constraints from Time Reversal Symmetry}
One says that $H:T\to Hermk(k)$ has a time reversal symmetry (TRS) if there is a pair consisting of an involution  $\tau$ on  $T$ and an anti--unitary operator $\Theta$ for which $\Theta^2=\pm 1$
such that $\Theta H(k) \Theta^*=H(\tau(k))$,  see \cite{Wigner,Notes}. As an anti--unitary operator there is a decomposition $\Theta=CU$ where $C$ is conjugation and $U$ is unitary.

Typical examples  are $T=T^n$, $\tau(k)=-k$, $\Theta=C$, $\Theta^2=1$ and $$\bar H(k)= H(-k)$$ Since $H$ is  Hermitian the pull back w.r.t.\ $\tau$ will be the identity on the Eigenvalue cover, in other words,
as is well--known, the full spectrum will be symmetric with respect to the involution, i.e.\ $\{\lambda_i(t)\}=\{\lambda_i(-t)\}$.
For the Eigenbundle geometry the symmetry implies that $\tau^*(\L_i)=\bar\L_i$ is the complex conjugate bundle, and hence has the negative Chern class of $\L_i$.
\begin{equation}
\label{involeq}
\tau^*c_1(\L_i)=c_1(\tau^* \L_i)=c_1(\bar\L_i)=-c_1(\L_i)
\end{equation}

\subsubsection{Example: Global constraints in the 3d torus case}

This allows us to add  to Theorem \ref{constrainthm}.

\begin{customthm}{\ref{constrainthm}}[Part 2]
If $H:T^3\to Herm(k)$ has TRS given as above and $\Tdeg$ is regular and in generic position then:
\begin{enumerate}
\item [(6)] $\chi_i(t)=-\chi_i(-t)$.
\item [(7)] The jumps at $t=0,\pi$ must be in $2\Z$.   Hence, if the local model is the spin model, the spin has to be integer. Furthermore the jumps are symmetric that is they go from $-\frac{1}{2}j_i(0)$, respectively $-\frac{1}{2}j_i(\pi)$, to $\frac{1}{2}j_i(0)$, respectively $-\frac{1}{2}j_i(\pi)$. In particular, if the  jump is $0$, then $\chi_i$ is $0$ as well in a neighborhood of $0$, respectively $\pi$.
\end{enumerate}
\end{customthm}
\begin{proof}
First notice that $\tau$ maps the slice at $t$ to the slice at $-t$. Now we can compute:
\begin{multline*}
\chi_i(t)= \int_{T^2}\phi_t^*(c_1(\L_i))=\int_{T^2}\tau^*\phi^*_{-t}(c_1(\L_i))=
\int_{T^2}\phi^*_{-t}(\tau^*c_1(\L_i))=\\\int_{T^2}\phi^*_{-t}(-c_1(\L_i))=-\chi_i(-t)
\end{multline*}
where we used  equation (\ref{involeq}).
By (5=6) $\chi_i(-\eps)=-\chi_i(\eps)_i$ and as $\pi\equiv -\pi \mod 2\pi$, $\chi_i(\pi-\eps)=-\chi_i(\pi+\eps)$ hence
 $j_i(0)=\chi(\eps)-\chi_i(-\eps)=2\chi_i(\eps)\in 2\Z$ and $j_i(\pi)=\chi_i(\pi+\eps)-\chi_i(\pi-\eps)=2\chi_i(\pi+\eps)$.
\end{proof}

\begin{cor}

On a 3d torus  family with time reversal symmetry:
\begin{enumerate}
\item One may not have a Weyl point with a coordinate $0$ or $\pi$. If there is a degenerate point with these coordinates, each degeneracy must be at least $3$.
Furthermore, if the singularity is of the  type of \S\ref{multcrosspar}, then it must be of integer spin.
\item For any singularity at $t$ there is a singularity with the same jump at $-t$:
$j_i(t_c)=j_i(-t_c)$.

\item The jumps satisfy $\sum_{0<t_c<\pi}2j_i(t_c)+j_i(0)+j_i(\pi)=0$.
\item If there is a singularity with local model of spin type, in the fiber over $t_c$ then there is the same local model in the fiber of  $-t_c$.
\end{enumerate}
\end{cor}

\begin{proof}

So we see that if there is a degenerate point with coordinate $0$ or $\pi$ it has to be at least a triple intersection.
Furthermore, we have $j_i(-t_c)=\chi_i(-t_c+\eps)-\chi_i(-t_c-\eps)=-\chi_i(t_c-\eps)+\chi_i(t_c+\eps)=j(t_c)$. This proves (2), and (3) then follows from (1) and (2).

For (4) we assume that there is a local model $H({\bk_0+\bx})-H(\bk_0)=\bx\cdot \bsigma=xS_x+yS_y+zS_z$ at $\bk_0=(t_1,t_2,t_c)$ then $H(-\bk_0+\bx))-H(-\bk_0))=\bar H({\bk_0-\bx})-\bar H(\bk_0)=-\bx\cdot \bar\bsigma=-(xS_x-yS_y+zS_z)$, where for the last line we have used the Zeeman basis for the representation in which $S_x,S_z$ have real coefficients and $S_y$ is purely imaginary. In the general case the complex conjugation adds an orientation reversal for $\bsigma$ and together with the sign of $\bx\to -\bx$  the total sign change is positive.
This is consistent with (2).

\end{proof}

For 2d families, we obtain the following version of (4)
\begin{cor}
\label{2dcor}
For a family on $T^2$ that has TRS $\bk\to -\bk$ and $\Theta=C$, if there is an equatorial Dirac point at $\bk_0$ then there is an equatorial Dirac point of opposite chirality at $-\bk_0$.
\end{cor}

\begin{proof}
By the same computation as above, we find that if $H(\bk_0+\bx)-H(\bk_0)=xS_x+yS_y$ then
$H(-\bk_0+\bx)-H(-\bk_0)=-(xS_x-yS_y)=-xS_x+yS_y$ and hence the chirality changes as now the sign of $\bx\to -\bx$ is $(-1)^2=1$.
\end{proof}

\section{Specific Examples}
Although our arguments so far have been totally general, a particular application we had in mind is
the application to the different quantum wire networks given by the honeycomb lattice and the lattices corresponding to the  P, D and G periodic minimal surfaces as  discussed in \cite{kkwk,kkwk3,kkwk4,kkwk2}.  In this setup one starts with a periodic graph and a periodic Harper Hamiltonian \cite{Harper,PST} and then constructs a family of Hamiltonians from it using Bloch theory. The latter can be encoded into a finite effective graph which has extra structures of a root and a spanning tree. We will give the effective graphs and the corresponding Hamiltonians and refer to the papers above for details.

Note that in this setup, there is a possibility to incorporate a magnetic field which makes the geometry non--commutative. This will be addressed in further research.

\subsection{Graph Examples}
The examples we considered are given by the effective graphs in Figure \ref{graphfig}.
\begin{figure}
\includegraphics[width=.7\textwidth]{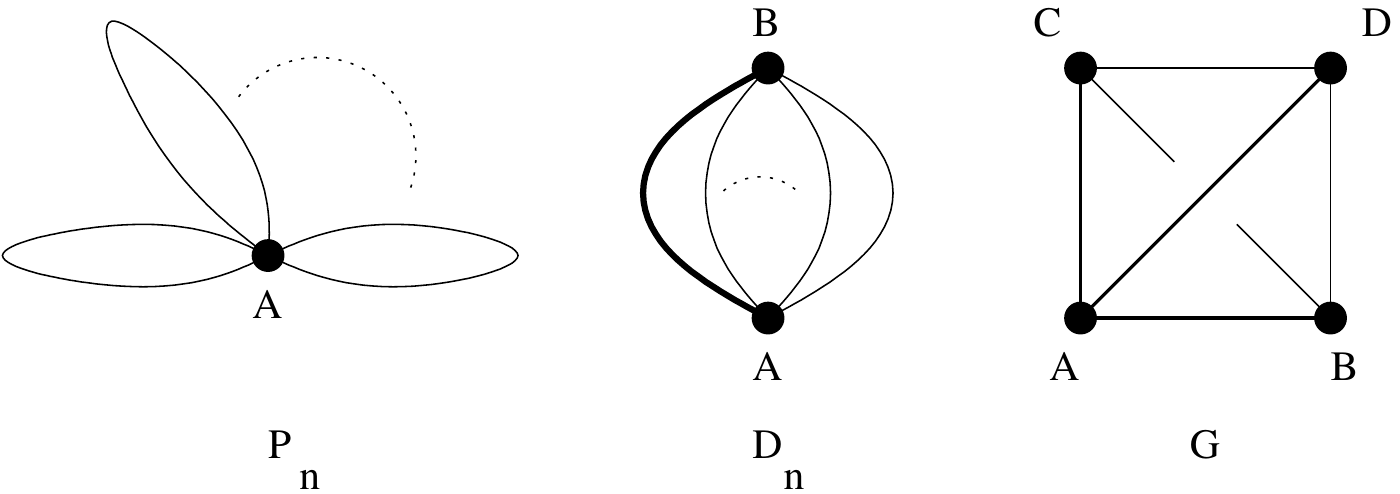}
\caption{\label{graphfig}
Graphs with  spanning trees and root $A$. The petal graphs $P_n$, with $n$ loops,
the digraphs $D_n$ with $n+1$ edges,
and the graph $G$}
\end{figure}
The dimension $d$ of the family is the number of non--spanning tree edges or the first Betti number of the graph. The family is  defined on $T^d$ and takes values in $Herm(k)$ where $k$ is the number of vertices.
The Hamiltonians are:
\begin{itemize}
\item[($P_n$)] $H(k_1,\dots,k_n)=\sum_{l=1}^n(e^{ik_l}+e^{-ik_l})$
\item[($D_n)$] $H(k_1,\dots,k_n)=\begin{pmatrix}0&1+\sum_{l=1}^ne^{ik_l}\\
1+\sum_{l=1}^ne^{-ik_l}&0
\end{pmatrix}$
\item[(G)] $H(k_1,k_2,k_3)=\begin{pmatrix}0&1&1&1\\
1&0&e^{ik_1}&e^{-ik_2}\\
1&e^{-ik_1}&0&e^{ik_3}\\
1&e^{ik_2}&e^{-ik_3}&0
\end{pmatrix}$
\end{itemize}
The $1$'s correspond to the spanning tree edges. All of these examples have TRS: $H(-{\bf k})=\bar H({\bf k})$.
They correspond to the following lattices.

\begin{enumerate}
\item $P_2$ corresponds to the square lattice, which is the simplest setting for the quantum Hall effect.
\item $P_3$ corresponds to the so--called primitive surface geometry.
\item $P_n$ in general is the geometry of a Bravais lattice.
\item $D_2$  corresponds to the honeycomb geometry (2d) which is the geometry of graphene.
\item $D_3$ which corresponds to the so--called diamond surface (3d).
 \item $G$  corresponds to the Gyroid geometry.
\end{enumerate}

The interesting three--dimensional cases from the point of view of the Eigenbundle geometry are the 3d cases $D_3$ and $G$, since $P_3$ has trivial Eigenbundle geometry as do all the $P_n$, where the Bloch bundle is just a trivial line
bundle. In all the examples, the Bloch bundle does not split into subbundles as all levels cross (or in the $P_n$ case there is only one level).

\subsection{The Honeycomb Lattice ($D_2$)}
\label{Graphenesec}
This is a two dimensional family on $T^2$. $\Tdeg$ are the two points
$(\rho_3,\bar \rho_3), (\bar \rho_3,\rho_3)$, $\rho=e^{i\frac{2\pi}{3}}$ at which there are the well known Dirac points of graphene, whose electronic properties are described by a Harper Hamiltonian: see the review \cite{CastroNeto} and references therein.

The local structure is well known and is given linearly  by the two--dimensional restriction $z=0$ of \S\ref{2x2case}. Here the  transformation matrix of the restricted version of Corollary \ref{doublecor} is $-1$ at $(\rho_3,\bar \rho_3)$ and $+1$ at $(\bar \rho_3,\rho_3)$ (see Appendix).
This allows to compute the Berry phase according to equation \eqref{berryeq}.
Notice that the two Dirac points have opposite chirality as dictated by Corollary \ref{2dcor}.

As $H^2(T_0)=0$
all the Chern charges vanish and the two Dirac points are in general not topologically stable.

\subsection{The Diamond ($D_3$)}
As computed in \cite{kkwk2},
$\Tdeg$ is given by the three circles on $T^3$ given by the equations $\phi_i=\pi, \phi_j\equiv \phi_k+\pi\; \mbox{mod}\; 2\pi$ with $\{i,j,k\}=\{1,2,3\}$.
 The singularities are double crossings of type $A_1$ but $\Tdeg$ is  not discrete and not contractible, hence not a regular case. Also, $\Tdeg$ is not smooth. There are singular points $(\pi,\pi,0)$, $(\pi,0,\pi)$ and $(0,\pi,\pi)$ where the three circles touch.
One can show that $T_0=T^3\setminus\Tdeg$ contracts onto a 1--dimensional CW--complex and hence
has $H^2(T_0)=0$. Thus there are no non--vanishing topological charges associated
to this geometry and no stability. Furthermore there is no slicing as any slice will hit $\Tdeg$.
Choosing a tubular neighbourhood of the smooth part of $\Tdeg$, we can define a function of Berry phases.
For this, one fixes a point in the smooth part of $\Tdeg$ and then chooses normal directions in the induced orientation. Then the family restricted to the two normal directions will be a restriction of the family \S\ref{2x2case} and like in the honeycomb case, computing the determinant of the matrix will yield the value of the Berry phase. By TRS symmetry for each point $k$ in the smooth part of $\Tdeg$ there is the opposite point $-k$ in the smooth part of $\Tdeg$ with opposite chirality.

The singular points of $\Tdeg$ are more complicated and will be the subject of further study.

\subsection{The Gyroid (G) }

For the gyroid the degenerate locus $\Tdeg$ is of real codimension 3 and consists of 4 points,
$(0,0,0)$, $(\pi,\pi,\pi)$, $(\frac{\pi}{2},\frac{\pi}{2},\frac{\pi}{2})$ and $(\frac{3\pi}{2},\frac{3\pi}{2},\frac{3\pi}{2})$.
The first two singular points correspond to  $(A_0,A_2)$ and $(A_2,A_0)$ singularities and the second two correspond to  an $(A_1,A_1)$ singularity, as calculated analytically in \cite{kkwk3}.
The latter furnish double Weyl points, i.e.\ two two--band crossings,
while the former yield three--band crossings.

Now $\Tdeg$ is the set of the four points above and $T_0=T^3\setminus \Tdeg$ contracts onto a 2--dim CW complex with non--trivial second homology \cite{momentum} and Theorem \ref{slicethm} applies.
All the charges are topologically stable.

The relevant numerics to compute the functions $\chi_i$ were carried out in \cite{kkwk5}.
As expected the $A_1$ singularities yield jumps by $\pm1$, and
the $A_2$ points yield jumps by $-2,0,2$ for the three bands that cross. This lead to the conjecture that
the latter points are also of spin type, which we now verify.

Namely, we add the local model description for all of these points and then show that one can use global constraints to completely describe the Eigenbundle geometry. A discussion of the behaviour of  the $A_2$ points under perturbations preserving some of the symmetry is given below.
The {\it prima vista} astonishing fact is that each of them splits into four $A_1$ points in compliance with the jumps given above.

\subsubsection{Extra Symmetry}
The Gyroid exhibits an extra symmetry given by $H(\bk+(\pi,\pi,\pi))=U^\dagger (-H(\bk))U$ with $U=diag(-1,1,1,1)$.
This means that the spectrum or Eigenvalue cover is invariant under simultaneously translating by $(\pi,\pi,\pi)$ and  flipping the sign of all Eigenvalues.
We see that if there is degeneracy at $\bk$ there is the same type of
degeneracy at $\bk+(\pi,\pi,\pi)$. Indeed this is true for the degeneracies listed above.
Moreover by Corollary \ref{negcor}, if the degeneracy is of spin type $(s_1,\dots,s_l)$ at $\bk$,  it is of spin type $(s_l,\dots, s_1)$ at $\bk+(\pi,\pi,\pi)$ with opposite chirality.
This adds information on the chirality of the double crossings.
Also, if (as we show) the $(A_2,A_0)$ singularity at zero is of spin type $(0,1)$ then necessarily we have that the $(A_0,A_2)$ singularity at $(\pi,\pi,\pi)$ is of spin type $(1,0)$ with opposite chirality.

\subsubsection{Local models}
For the two points $(A_1,A_1)$, we know that the local models are given by the usual double crossing ${\bf x}\cdot {\bsigma}$
for spin $\frac{1}{2}$. That is they are of spin type $(\frac{1}{2},\frac{1}{2})$, but the chirality remains to be determined.
For the $A_2$ singularity there could be a choice of local models.
Using perturbation theory, we computed  the local models.
The result is:
\begin{prop}
\label{localstrucprop}
The local models for the Gyroid are as follows.
\begin{enumerate}
\item The point $(0,0,0)$ is of spin type $(1,0)$ with the chirality $1$.
\item The point $(\pi,\pi,\pi)$ is of spin type $(0,1)$ with chirality $-1$.
\item The point $(\frac{\pi}{2},\frac{\pi}{2},\frac{\pi}{2})$ is of spin type $(\frac{1}{2},\frac{1}{2})$ with chirality $(-1,1)$.
\item The point $(\frac{3\pi}{2},\frac{3\pi}{2},\frac{3\pi}{2})$ is of spin type $(\frac{1}{2},\frac{1}{2})$ with chirality $(-1,1)$.
\end{enumerate}
\end{prop}
\begin{proof}
The computation for the point $(0,0,0)$ is done in detail in the Appendix. The extra symmetry then implies the result for the point $(\pi,\pi,\pi)$.
The computation for the chirality of the point $(\frac{\pi}{2},\frac{\pi}{2},\frac{\pi}{2})$  is also in the Appendix.
It fixes the chirality of
 $(\frac{3\pi}{2},\frac{3\pi}{2},\frac{3\pi}{2})$. Alternatively the chiralities of the double crossings follow from the global analysis below.
\end{proof}

A schematic version is given in Figure \ref{schematicfig}.
From this one can read off the entire functions $\chi_i$ using the arguments of Lemma \ref{calclem}. The results are in Figure \ref{chivals}.
As a preview, we discuss the highest level $\chi_4$. Since there is no jump at $0$, we have that $\chi_4=0$ in the invervals adjacent to zero.
At $\pi/2$, $\chi_4$  jumps up by one, as the chirality of the top Weyl point is positive. At $\pi$ it jumps down by two, since the spin $1$ chirality is $-1$ and the top band then jumps by $-2$. At $3\pi/2$, $\chi_4$ jumps up by one again, to yield a net jump of $0$.

\begin{figure}
\includegraphics[width=.8\textwidth]{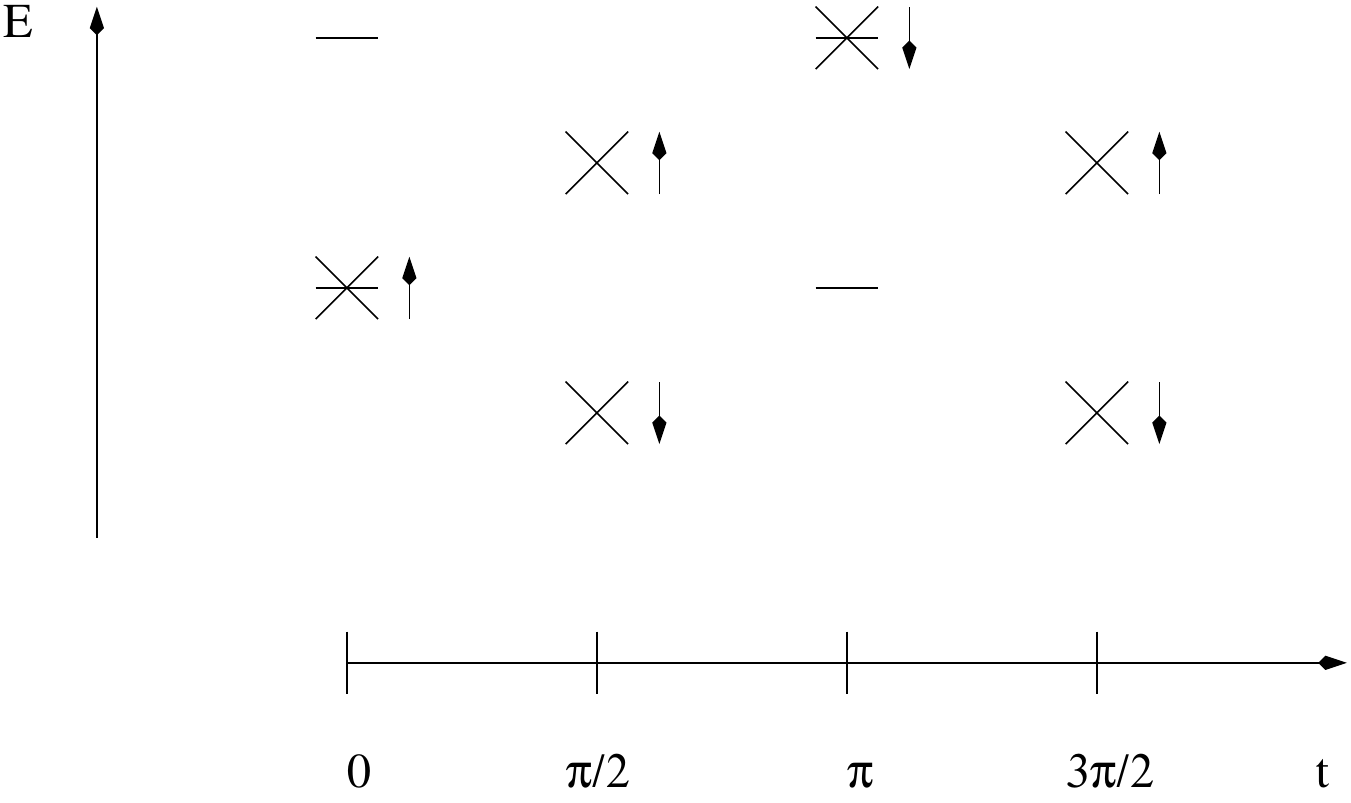}
\caption{\label{schematicfig} Schematic of the singularities for the $z$ slicing. Single lines are of type $A_0$, i.e.\ no crossing. Crosses indicate $A_1$ Weyl points. These are spin $1/2$. The $A_2$ triple crossings are of spin 1 type. The chiralities are indicated by arrows. $\uparrow$ means $+1$ and $\downarrow$ means $-1$ chirality. The axes are the slicing parameter $t$ and the energy $E$. The latter is only schematic, to indicate the relative positions of the level.}
\end{figure}
\subsubsection{Global analysis}
Since the Gyroid has time reversal symmetry, only one of the chiralities needs to be computed.
In fact, knowing the degeneracies are of spin type and their location, the functions $\chi_i$ are determined up to an overall change of sign. This is fixed by one chirality. The fact that the triple crossings are indeed of spin type is a separate proof,  however.

Let $t_c^0=0,t_c^1=\pi/2, t_c^2=\pi, t_c^3=3\pi/2$ be the critical slice parameters.
Pick intermediate parameters $0=t_c^0<t_1<t_c^1< \dots <t_4<2\pi$. We may choose  $t_4=-t_1,t_3=-t_2$.

To illustrate the power of Theorem \ref{constrainthm}, we give the details.
\begin{lem}
\label{calclem}
Due to TRS and the extra symmetry,
the chiralities of the $(\frac{1}{2},\frac{1}{2})$ spin type points are fixed by the chirality of one of the double crossings. Given this chirality, $\chi_1$ and $\chi_4$ and $\chi_2,\chi_3$ are fixed up to a parameter.

Adding that one of the triple crossings is of spin type, the other is as well. The spin type has to be spin $1$ and all the functions $\chi_i$ and chiralities are fixed by fixing one of the chiralities of either one of the spin--1 triple crossings or one of the spin--$\frac{1}{2}$ double crossings.
\end{lem}
\begin{proof}

The family is time reversal invariant, so that $\chi_i(t_1)=-\chi_i(t_4)$
and $\chi_i(t_2)=-\chi_i(t_3)$.
Thus it suffices to know the $\chi_i(t_k)$ for $i=1,2,4;k=1,2$ to know
the whole step functions $\chi_i$. The function $\chi_3$ can  be computed by Theorem \ref{constrainthm} (2).

Assume that $j_1(\pi/2)=1$ then by TRS $j_1(3\pi/2)=1$ and by the extra symmetry $j_4(3\pi/2)=1$, which in turn means by TRS means that $j_4(\pi/2)=1$. We could have equally started with any one of these four chiralities. This fixes the following data.
\begin{eqnarray*}
\chi_1(t_2)-\chi_1(t_1)=j_1(\pi/2)=+1&&\chi_1(t_4)-\chi_1(t_3)=j_1(3\pi/2)=+1\\
\chi_2(t_2)-\chi_2(t_1)=j_2(\pi/2)=-1&&\chi_2(t_4)-\chi_2(t_3)=j_2(3\pi/2)=-1\\
\chi_3(t_2)-\chi_3(t_1)=j_3(\pi/2)=-1&&\chi_3(t_4)-\chi_3(t_3)=j_3(3\pi/2)=-1\\
\chi_4(t_2)-\chi_4(t_1)=j_4(\pi/2)=+1&&\chi_4(t_4)-\chi_4(t_3)=j_4(3\pi/2)=+1
\end{eqnarray*}

We also know that $j_4(0)=j_1(\pi)=0$, since the respective Eigenvalues are not degenerate at these points and hence by Theorem \ref{constrainthm} Part II, it follows that $\chi_4(t_1)=\chi_4(t_4)=\chi_1(t_2)=\chi_1(t_3)=0$.
Thus we know the full functions $\chi_1,\chi_4$. We also know that $j_1(0)=-2$ and $j_4(\pi)=-2$. Hence if either of the  $A_2$ singularity at $0$ is of spin type, it is of spin type 1 and the other has to be of spin type 1 as well due to the extra symmetry. The chirality is also fixed to be $+1$ at $(0,0,0)$ and $-1$ at $(\pi,\pi,\pi)$.
Furthermore assuming spin type, we see that $j_2(0)=j_3(0)=0$ and again the full functions are fixed.
The extra condition of $j_3(0)=j_2(\pi)=2$ is then automatically satisfied. Knowing the chirality of one of the $A_2$ singularities fixes the jumps at $0$ and $\pi$ and hence the chirality of the double crossings via Theorem \ref{constrainthm} (2) and TRS symmetry.

\begin{figure}
\raisebox{.5in}{\begin{tabular}[b]{l|rrrr}
&$t_1$&$t_2$&$t_3$&$t_4$\\
\hline
$\chi_1$&-1&0&0&1\\
$\chi_2$&0&-1&1&0\\
$\chi_3$&1&0&0&-1\\
$\chi_4$&0&1&-1&0\\
\end{tabular}}
\quad
\includegraphics[width=.6\textwidth]{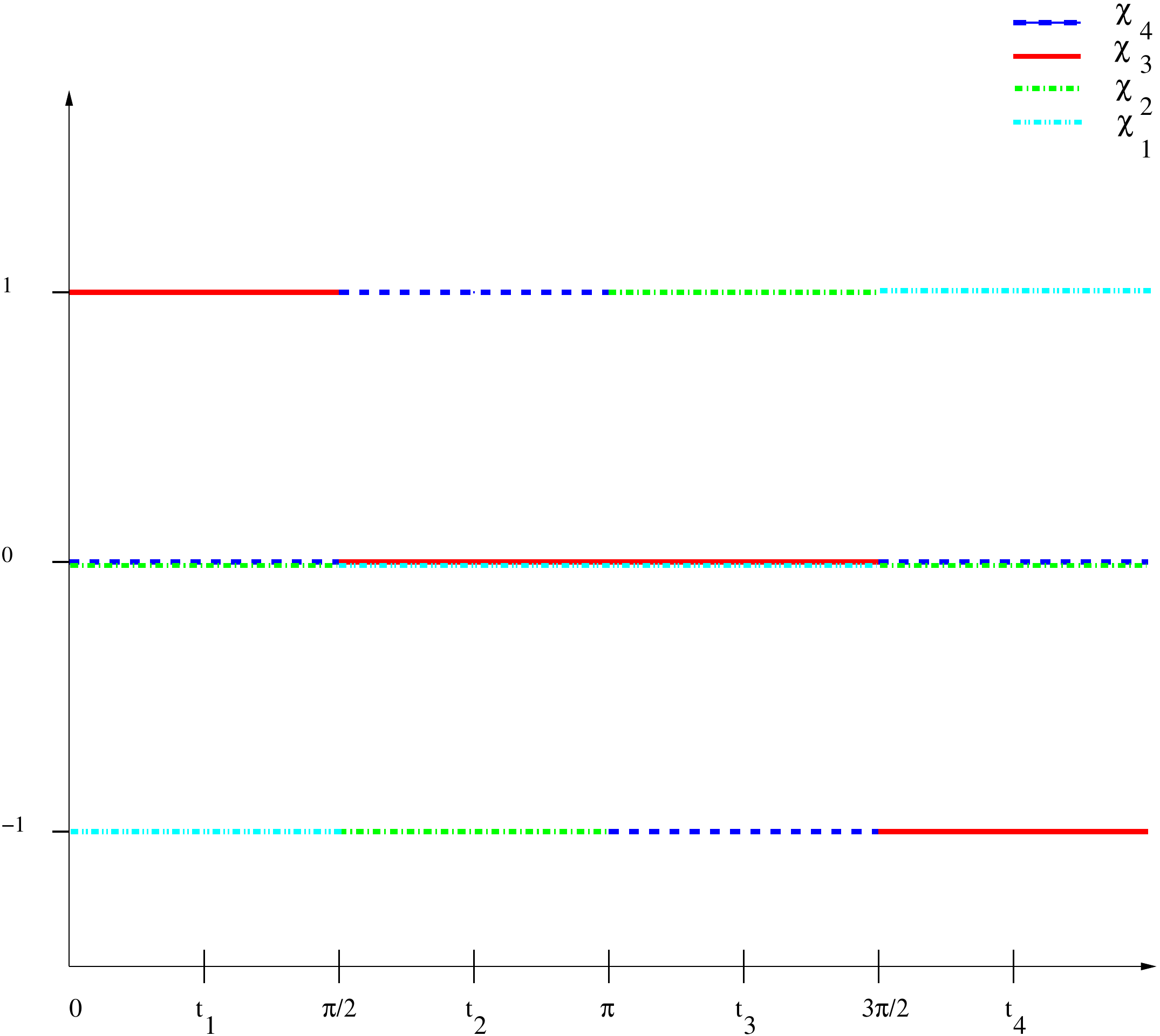}
\caption{\label{chivals}Values of the functions $\chi_i$ and graphs}
\end{figure}

On the other hand, if we do not assume that one of the $A_2$ singularities is of spin type, we can still use Theorem \ref{constrainthm} (3) and (4) to obtain the equations:
$j_2(0)+j_2(\pi)=2,j_3(0)+j_3(\pi)=2, j_2(0)+j_3(0)=2,j_2(\pi)+j_3(\pi)=0$.
We can then further reduce to one parameter, say $j_3(0)=m\in 2 \Z$, then $j_3(\pi)=j_2(0)=2-m$ and $j_2(\pi)=m$.
These automatically satisfy Theorem \ref{constrainthm} (2).

Changing the chirality flips all signs in the argument.

\end{proof}

\begin{prop}
The functions $\chi_i$ for the Gyroid and the slicing
$\phi_t:(\theta_1,\theta_2)=(\theta_1,\theta_2,t)$ are given by the table in Figure \ref{chivals} . 
\end{prop}
\begin{proof}
By the Lemma all we need to know if one of the chiralities of Proposition \ref{localstrucprop}.
\end{proof}

Also note that the singularity in the fiber over $(\pi/2,\pi/2,\pi/2)$ gives rise to another singularity in the fiber at $(-\pi/2,-\pi/2,-\pi/2)$ by both TRS and the extra symmetry.
This forces  another singularity somewhere else as the following computation shows.

\subsection{Deformation under symmetry}
If we deform the Hamiltonian in the system above to resolve the triple crossing into normal double crossing singularities, but keep the time reversal symmetry,
we know:
\begin{enumerate}
\item There will be no singularities at $t=0,\pi$ as these would have to be at least triple crossings.
\item Isolated double crossings will appear pairwise. For every double crossing at $\pi-t$ that appears in a small neighborhood   of $\pi$ there will be a corresponding double crossing at $\pi+t$ with opposite jumps.
\item If all the double crossings are between $t_2^*$ and $t_3^*= -t_2^*$, then the total jumps  between $t^*_2$ and $t_3^*$ are by $2,0,-2$.
\item For $\chi_i$ to jump by two, the corresponding Eigenvalue will have to cross two times with the same sign.
\item If $\chi_i$ jumps by $1$ at some point, by time reversal symmetry, it jumps a second time by $1$.

\end{enumerate}
Looking at these constraints, we see that the minimal resolution will have to have 4 double crossings and this is borne out by the numerics \cite{kkwk5}. More precisely, if band $4$ crosses with $3$ and has a jump of $-1$ at $t_2^*<\pi-s_1<\pi$ then due to time reversal symmetry they will cross again at $\pi+s_1$ with another jump of $-1$ and thus have a net jump of $-2$ for the band $4$ from $t_2^*$ to $t_3^*$ as needed. Likewise, if the bands $2$ and $3$ cross at $t_2^*<\pi-s_2<\pi$ with a jump of band $2$ by $1$ in the Chern number, then there will be a second crossing with the same jump at $\pi+s_2$ and these two will add up to a net jump by $2$ for band $2$.
There is no further crossing needed as the band $3$ will have a total jump of $2-2=0$.

What is not determined is if $s_1 \geq s_2$  or $s_1\leq s_2$. In fact the order of $s_1$ and $s_2$ may well be different for different deformations; equality is not generic.

\section*{Acknowledgments}
RK thankfully acknowledges
support from NSF DMS-0805881 and DMS-1007846.
BK  thankfully acknowledges support from the  NSF under the grant  PHY-1255409.
 Any opinions, findings and conclusions or
recommendations expressed in this
 material are those of the authors and do not necessarily
reflect the views of the National Science Foundation.
This work was partially supported by grants from the Simons Foundation (\#267481 to Erika Birgit Kaufmann and \#267555 to Ralph Kaufmann).
Both RK and BK thank the Simons Foundation for this support.
They also thank M.\ Marcolli for insightful discussions.

Parts of this work were completed when RK was visiting
the IHES in Bures--sur--Yvette and the University of Hamburg with a Humboldt fellowship and RK and BK were visiting the Institute for Advanced Study in Princeton and the Max--Planck--Institute in Bonn. They gratefully acknowledge
all this support.

\appendix

\section{Calculations}
In this appendix, we give some of the calculations.
\subsection{Honeycomb/graphene}
We expand $$
\begin{pmatrix}
0&1+e^{ik_1}+e^{ik_2}\\
1+e^{-ik_1}+e^{-ik_2}&0
\end{pmatrix}$$
at $\bk_0=(\frac{2\pi}{3},-\frac{2\pi}{3})$
and obtain
$$H(\bk_0+\bx)=
\begin{pmatrix}
0&0\\0&0
\end{pmatrix}
+\begin{pmatrix}
0&-\frac{\sqrt{3}(x-y)-i(x+y)}{2}\\
-\frac{\sqrt{3}(x-y)+i(x+y)}{2}&0\\
\end{pmatrix}  +O(\bx^2)
$$
 Comparing with \eqref{Weylfameq} and keeping in mind that $\bS=\frac{1}{2}\bsigma$,
we can read off the transformation $a=-\sqrt{3}(x-y), b=x+y$ which has negative determinant and chirality.
This is indeed equatorial Dirac, since the diagonal entries are $0$.

Expanding at  $\bk_0=(-\frac{2\pi}{3},\frac{2\pi}{3})$ yields
$$H(\bk_0+\bx)=
\begin{pmatrix}
0&0\\0&0
\end{pmatrix}
+\begin{pmatrix}
0&\frac{\sqrt{3}(x-y)-i(x+y)}{2}\\
\frac{\sqrt{3}(x-y)+i(x+y)}{2}&0\\
\end{pmatrix}  +O(\bx^2)
$$
and the transformation $a=\sqrt{3}(x-y), b=x+y$ which has positive determinant and chirality.

\subsection{Gyroid}

\subsubsection{Triple crossings}
We compute that the $A_2$ singularity at $(0,0,0)$ for the Gyroid is
of spin 1 type and has positive chirality.

To compute the local model we used first order perturbation theory and expanded $H$ near
$\bk_0=(0,0,0)$ as $H(\bk_0+\bx)=H(\bk_0)+H_1(\bf{x})+{\mathcal O}({\bf x}^2)$. This yields
$$
H(0,0,0)=\begin{pmatrix}
0&1&1&1\\
1&0&1&1\\
1&1&0&1\\
1&1&1&0
\end{pmatrix}\quad
H_1(\bx)=
\begin{pmatrix}
0&0&0&0\\
0&0&ix&-iy\\
0&-ix&0&iz\\
0&iy&-iz&0\\
\end{pmatrix}$$
Since, we is a triple degeneracy for the Eigenvalue $-1$, we have to transform to a unitary basis to do the projection.
The transformation matrix to diagonal form $diag(3,-1,-1,-1)$ of $H(0,0,0)$ is
$$U=\begin{pmatrix}
\frac{1}{2}&\frac{1}{\sqrt{2}}&0&\frac{1}{2}\\
\frac{1}{2}&-\frac{1}{\sqrt{2}}&0&\frac{1}{2}\\
\frac{1}{2}&0&\frac{1}{\sqrt{2}}&-\frac{1}{2}\\
\frac{1}{2}&0&-\frac{1}{\sqrt{2}}&-\frac{1}{2}\\
\end{pmatrix}$$
We then have the projection to the ${-1}$ Eigenspace $PU^{\dagger}H_1(\bx)UP$ where $P=diag(0,1,1,1)$.
The resulting block of the matrix acting  in the subspace with eigenvalue $-1$  is
$$
\begin{pmatrix}
0&-\frac{i}{2}(x+y)&\frac{i}{2\sqrt{2}}(x-y)\\
\frac{i}{2}(x+y)&0&-\frac{i}{2\sqrt{2}}(x+y+2x)\\
-\frac{i}{2 \sqrt{2}}(x-y)&\frac{i}{2\sqrt{2}}(x+y+2x)&0\\
\end{pmatrix}
$$

Setting $a=\frac{1}{2\sqrt{2}}(x+y+2z)=\frac{1}{2\sqrt{2}}(x-y),c=\frac{1}{2}(x+y)$ the matrix takes the form
$$
i\begin{pmatrix}
0&-a&b\\
a&0&-c\\
-b&c&0
\end{pmatrix}=aiL_x+biL_y+ciL_z=(a,b,c)\cdot \tilde\bS$$
where $L_x,l_y,L_z$ are the standard generators for $so(3)$ and $\tilde \bS$ is the corresponding spin representation. This is not in the standard form, but for the chirality, we only need to determine the sign of the transformation $T:(x,y,z)\to (a,b,c)$
$$
sign(det(T))=sign(\begin{vmatrix} \frac{1}{2\sqrt{2}}&\frac{1}{2\sqrt{2}}&\frac{1}{2}\\
\frac{1}{2\sqrt{2}}&-\frac{1}{2\sqrt{2}}&0\\
\frac{1}{2}&\frac{1}{2}&0\end{vmatrix})=sign(\frac{1}{4})=+1
$$

\subsubsection{Weyl points}
At the point $(\pi/2,\pi/2,\pi/2)$ the matrices are
$$
H(\frac{\pi}{2},\frac{\pi}{2},\frac{\pi}{2})=
\begin{pmatrix}
0&1&1&1\\
1&0&i&-i\\
1&-i&0&i\\
1&i&-i&0
\end{pmatrix}\quad
H_1(\bx)= \begin{pmatrix}
0&0&0&0\\0&0&-x&-y\\0&x&0&-z\\0&y&z&0
\end{pmatrix}
$$
The transformation matrix is
$$U=\begin{pmatrix}
\frac{1}{6}(-3-i\sqrt{3})&-\frac{1}{\sqrt{6}}&\frac{1}{6}(3-i\sqrt{3})&\frac{1}{\sqrt{6}}\\
\frac{1}{6}(3i+\sqrt{3})&-\frac{i}{\sqrt{6}}&\frac{1}{6}(-3i+\sqrt{3})&\frac{i}{\sqrt{6}}\\
0&\frac{1}{\sqrt{2}}&0&\frac{1}{\sqrt{2}}\\
\frac{1}{\sqrt{3}}&\frac{i}{\sqrt{6}}&\frac{1}{\sqrt{3}}&-\frac{i}{\sqrt{6}}
\end{pmatrix}
$$
this yields the following $2\times 2$ matrices for the Eigenspaces $-\sqrt{3}$.
$$\begin{pmatrix}
-\frac{y}{3}&\frac{-\sqrt{3}x-\sqrt{3}y-2\sqrt{3}z+i(3x+y)}{6\sqrt{2}}\\
\frac{-\sqrt{3}x-\sqrt{3}y-2\sqrt{3}z-i(3x+y)}{6\sqrt{2}}
\end{pmatrix}$$

For the transformation $a=\frac{-x-y-2z}{\sqrt{6}},b=\frac{3x+y}{3\sqrt{2}},c=-\frac{2y}{3}$ this becomes
$\frac{1}{2}\begin{pmatrix}
c&a-ib\\
a+ib&c
\end{pmatrix}
=aS_x+bS_y+cS_z
$ which yields the chirality $-1$. For Eigenspace $\sqrt{3}$ the $2\times 2$ matrix is the complex conjugate of the matrix above and
the transformation is accordingly $a=\frac{-x-y-2z}{\sqrt{6}},b=-\frac{3x+y}{3\sqrt{2}},c=-\frac{2y}{3}$ which yields the opposite chirality $1$.

\end{document}